\documentclass[a4paper,11pt,dvips]{article}

\usepackage{amsfonts}
\usepackage{amsmath}
\usepackage{amsthm}
\usepackage[auth-sc,affil-sl]{authblk}
\usepackage{bm}
\usepackage{multirow}
\usepackage{verbatim}
\usepackage{enumitem}

\renewcommand
\baselinestretch{1.0} 
\usepackage{url}
\usepackage{natbib}
\bibpunct{(}{)}{,}{a}{}{,}
\usepackage{graphicx}
\usepackage{lscape}
\usepackage[usenames]{color}
\usepackage{sansmath}
\usepackage{subfigure}
\usepackage{rotating}

\DeclareMathAlphabet{\mathpzc}{OT1}{pzc}{m}{it}

\newcommand\independent{\protect\mathpalette{\protect\independenT}{\perp}}
\def\independenT#1#2{\mathrel{\rlap{$#1#2$}\mkern2mu{#1#2}}}

\def\Expect{\mathbb{{E}}}

\newtheorem{theorem}{Theorem}%[chapter]
\newtheorem{definition}[theorem]{Definition}%[chapter]
\newtheorem{lemma}[theorem]{Lemma}%[chapter]
\theoremstyle{remark}
\newcommand \address[1]{\gdef \@address{#1}}
\makeatletter
\long\def\@footnotetext#1{\insert\footins{\def\baselinestretch{1.2}\footnotesize
\interlinepenalty\interfootnotelinepenalty
\splittopskip\footnotesep \splitmaxdepth \dp\strutbox
\floatingpenalty \@MM \hsize\columnwidth \@parboxrestore
\edef\@currentlabel{\csname
p@footnote\endcsname\@thefnmark}\@makefntext
{\rule{\z@}{\footnotesep}\ignorespaces #1\strut}}}

\long\def\symbolfootnote[#1]#2{\begingroup%
\def\thefootnote{\fnsymbol{footnote}}\footnote[#1]{#2}\endgroup}

\providecommand{\keywords}[1]{\textit{Keywords:} #1}

\def\maketitle{%
  \null
  \thispagestyle{empty}%
  \begin{center}\leavevmode
    \normalfont
    {\LARGE \bf \@title\par}%
    {\normalsize \@author\par}%
    \vskip 0.05 cm
    \vskip 0.05cm
    {\normalsize \@date\par}%
  \end{center}%
}

\makeatletter
\newcommand{\institute}[1]{\newcommand{\@institute}{#1}}

\renewcommand\textsl{\textcolor{blue}}

\begin{document}
\title{Multiply robust dose-response estimation for multivalued causal inference problems}
\author[1]{Cian Naik}
\author[1]{Emma J. McCoy}
\author[2]{Daniel J. Graham}
\affil[1]{Department of Mathematics, Imperial College London, London, UK}
\affil[2]{Corresponding author: Department of Civil Engineering, Imperial College London, London, SW7 2AZ, UK. Email: \texttt{d.j.graham@imperial.ac.uk}} 
\date{}

\maketitle
\begin{abstract}
This paper develops a multiply robust (MR) dose-response estimator for causal inference problems involving multivalued treatments. We combine a family of generalised propensity score (GPS) models and a family of outcome regression (OR) models to achieve an average potential outcomes estimator that is consistent if just one of the GPS or OR models in each family is correctly specified. We provide proofs and simulations that demonstrate multiple robustness in the context of multivalued causal inference problems. 
\end{abstract}
.

\section{Introduction}

In recent years a great deal of attention has been devoted to doubly robust (DR) estimators for causal inference which combine propensity score (PS) and outcome regression (OR) models to give an average treatment effect (ATE) estimator that is consistent and asymptotically normal under correct specification of just one of the two component models. The key advantage of the DR approach is that it affords the analyst two chances at correct inference. DR estimators were first introduced by \citet{Scharfstein/et/al:1999} and have been studied extensively in theory and application in the binary treatment setting \citep[e.g.][]{Robins:2000,Robins/et/al:2000c,Robins/Rotnitzky:2001,VanDerLaan/Robins:2003,Lunceford/Davidian:2004,Bang/Robins:2005,Kang/Schafer:2007}. A small number of papers have formulated DR models for dose-response estimation \citep[e.g.][]{Flores/Mitnik:2009,Zhang/et/al:2012,Graham/et/al:2016}.

The recent missing data literature has developed multiply robust (MR) estimation approaches, which consistently estimate a population mean, given ignorable missingness, if just one of multiple combined models is correctly specified. Since causal inference methods essentially treat unobserved potential outcomes as missing data, there are opportunities for translation of missing data methods in the causal inference setting. In this paper we generalise and apply the binary MR missing data estimation approach to address causal inference problems for dose-response estimation involving multivalued treatments. 

In so doing we aim to make two main contributions to the causal inference literature. First, we advance methods for dose-response estimation based on the generalised propensity score (GPS) \citep[e.g.][]{Imbens:2000,Hirano/Imbens:2004}, which have been understudied relative to binary treatment problems. Second, we contribute to the exiting literature on robust modelling for causal inference by demonstrating how an MR model can be formulated and applied in the causal setting.  

Building on the work of \citet{Han/Wang:2013}, \citet{Han:2014a} and \citet{Han:2014b} we formulate our MR approach by combining a family of GPS models, $\mathcal{P}$, and a family of OR models, $\mathcal{A}$, to construct an average potential outcome (APO) estimator that is consistent if either $\mathcal{P}$ contains a correctly specified PS model or $\mathcal{A}$ contains a correctly specified OR model. We provide proofs of MR properties for APO estimation of multivalued treatments and we present simulations which demonstrate that the approach is robust to problems of confounding and functional form misspecification which are routinely encountered in causal analyses.   

The paper is structured as follows. Section two briefly outlines the theory of doubly robust estimation. Section three defined the concept of multiply robust estimation in the context of causal inference problems involving multivalued treatments, provides proofs of MR properties for estimation of average potential outcomes (APOs), and proposes an algorithm for numerical implementation. Simulation results are then presented in section four. Conclusions are drawn in the final section.     

\section{Doubly Robust Estimation}

Causal inference models are designed to estimate the effect of a treatment on an outcome of interest. Typically, the set up is one in which the treatment is not randomly assigned, and consequently, confounding may be present. Let $z_i = (y_i,d_i,x_i)$, $i=(1,...,n)$, be a random vector of observed data available for causal modelling, where for the $i$-th unit of observation $y_i$ denotes an outcome (or response), $d_i$ the treatment received, and $x_i$ a vector of pre-treatment covariates that represent confounding. The treatment could be binary, taking values in $\mathcal{D} \in\{0,1\}$; or multivalued, with $Q$ treatment categories $\mathcal{D}\equiv(d_1,d_2,...,d_Q)$. We use lower case $d_q$ to denote specific potential, as distinct from observed, treatments.  

For every level of treatment $d_q$ we can define a potential outcome, $Y_i(d_q)$, with $\mathcal{Y}_i=\{Y_i(d_q): d_q \in \mathcal{D}\}$ denoting the full set of potential outcomes. The estimands of interest in causal inference studies include average potential outcomes (APOs), $APO=\Expect[Y(d_q)]$, which measure the average response across the population under treatment level $D_i=d_q$; or average treatment effects (ATEs), $ATE=\Expect[Y(d_q)-Y(0)]$, which measure the difference in average response that would occur under treatment levels $D_i=d_q$ and $D_i=0$. 

To estimate a `causal' APO or ATE we must adjust for relevant unit level confounding characteristics in order to capture the marginal effect of the treatment. Identification of causal estimands requires three key assumptions to hold. First, responses must be conditionally independent of treatment assignment given observed covariates; i.e. $(Y_i(0),Y_i(1)) \independent I_{1}(D_i)|X_i,$ in the case of binary treatments, and $Y_i(d_q) \independent  I_{d_q}(D_i)|X_i$ for all $d_q \in \mathcal{D}$ in the case of multivalued treatments; where $I_{d_q}(D_i)$ is the indicator function for receiving treatment level $d_q$. Second, conditional on covariates $X_i$, the probability of assignment to treatment must be strictly positive for all $x$ and $d$, or at least this must hold within some region of treatment, $\mathcal{C}\subseteq \mathcal{D}$ say, referred to as the common support region. Third, the relationship between observed and potential outcomes must comply with the Stable Unit Treatment Value Assumption (SUTVA)\citep[e.g.][]{Rubin:1978}, which requires that the observed response under a given treatment allocation is equivalent to the potential response under that treatment allocation. For binary treatments the SUTVA implies $Y_i=I_{1}(D_i)Y_i(1)+ (1-I_{1}(D_i))Y_i(0)$ for all $i=1,...,N$, and for multivalued treatments $Y_i \equiv I_{d_q}(D_i)Y_i(d_q)$ for all $d_q \in \mathcal{D}$, for all $Y_i (d_q) \in \mathcal{Y}_i$, and for $i=1,...,N$. 

If these three conditions hold then the estimands of interest are identified and estimation may proceed via a number of different approaches \citep[for a review see][]{Imbens/Rubin:2015}. Two are of particular interest in this paper. First, we can proceed by positing an outcome regression (OR) model $\Psi^{-1}\{m(X_i,D_i;\beta)\}$ for the mean response $\Expect(Y_i|D_i,X_i)$, where $\Psi$ is a known link function, $m()$ is a regression function, and $\beta$ is an unknown parameter vector. If the OR model is correctly specified APOs and ATEs can be consistently estimated by averaging over the covariates for given values of $d_q$. Second, we can assume a model for $f_{D|X}(d_i|x_i)$, the conditional density of the treatment given the covariates, and use this model to estimate propensity scores (PSs), which we denote $\widehat{\pi}(D_i|X_i; \widehat{\alpha})$ for unknown parameter vector $\alpha$. PS weighting estimators of the form attributed to \citet{Horovitz/Thompson:1952} can then be used to estimate ATEs consistently if the PS model is correctly specified.
 
Doubly robust (DR) estimation combine OR and PS models in such a way that consistent estimates of APOs and ATEs can be obtained if just one of two component models are correctly specified. This can be achieved by weighting the OR model with a function of the inverse PS values. Using the notation of \citet{Lunceford/Davidian:2004},  the DR estimator for binary treatment APOs can be written
\begin{align} \label{eq1}
\hat{\mu}_{DR}(d_q)=\frac{1}{n}\sum_{i=1}^n\bigg[\frac{I_{d_q}(D_i)Y_i}{h(d_q|X_i;\alpha)}-\frac{I_{d_q}(D_i)-h(d_q|X_i;\alpha)}{h(d_q|X_i;\alpha)}\Psi^{-1}\{m(X_i,d_q;\hat{\beta})\}\bigg]
\end{align}
where 
\begin{align*}
h(I_{d_q}(D_i)|X_i;\alpha)=I_{d_q}(D_i)\hat{\pi}(D_i|X_i;\hat{\alpha})+[1-I_{d_q}(D_i)]\{1-\hat{\pi}(D_i|X_i;\hat{\alpha})\}
\end{align*}
A DR estimator for the ATE follows as $\hat{\tau}_{DR}(1)=\hat{\mu}_{DR}(1)-\hat{\mu}_{DR}(0)$ \citep[for a proof see][]{Lunceford/Davidian:2004}. 

For multivalued treatments, we can extend the binary approach as follows 
\begin{align} \label{eq2}
\hat{\mu}_{DR}(d_q)=\frac{1}{n}\sum_{i=1}^n\bigg[\frac{I_{d_q}(D_i)Y_i}{\hat{\pi}(d_q|X_i;\hat{\alpha})}-\frac{I_{d_q}(D_i)-\hat{\pi}(d_q|X_i;\hat{\alpha})}{\hat{\pi}(d_q|X_i;\hat{\alpha})}\Psi^{-1}\{m(X_i,d_q;\hat{\beta})\}\bigg]
\end{align}
Where $\hat{\pi}(d_q|X_i;\hat{\alpha})$ is a model for 
\begin{align*}
\pi(d_q|X_i)=\mathbb{P}(D_i=d_q|X_i)=\mathbb{E}(I_{d_q}(D_i)|X_i)
\end{align*}
which we form by assuming a p.m.f. $f_{D|X}(d_q|x_i,\alpha)$ for the GPS $\pi(d_q|x_i)$, and estimating the parameter $\hat{\alpha}$ from a regression model using the observed treatment doses $D_i$ and covariates $X_i$ \citep[for details see][]{Imbens:2000}. A proofs of the DR properties of the multivalued dose-response estimator shown in equation (\ref{eq2}), which to our knowledge has not previously appeared  in the literature, is given in the appendix. 

\section{A multiply robust dose-response estimator}
Having discussed DR estimators, the question naturally arises as to whether estimators can be constructed which are `more' than DR, and what this `multiple' robustness would entail. A number of recent papers have explored the issue of multiple robustness in the context of missing data problems including \citet{Chan:2013} and \citet{Chan/Yam:2014}. Here we build on the work of \citet{Han/Wang:2013}, \citet{Han:2014a} and \citet{Han:2014b} which develop the empirical likelihood based approach of \citet{Qin/Zhang:2007} to construct an MR estimation approach for missing data that consistently estimate a population mean given ignorable missingness. We generalise this work to the multivalued setting and derive a type of MR estimator for causal inference dose-response estimation problems in the following sense.
\begin{definition}(Multiple-robustness) We postulate multiple PS models $\mathcal{P}=\{\pi^j(x;\alpha^j):j=1,...,J\}$ and multiple OR models $\mathcal{A}=\{a^k(x,d;\beta^k):k=1,...K\}$. An estimator of an ATE or APO combining all of these models is multiply robust if it is consistent when either $\mathcal{P}$ contains a correctly specified model for the propensity score or $\mathcal{A}$ contains a correctly specified outcome regression model. 
\end{definition}
We develop our approach for multivalued treatments noting that results for binary treatments follow directly as a special case.

For a multivalued treatment in a causal inference problem, we denote the number of observed responses with treatment level $d_q$ as $m_q$ with $\mathcal{M}_q=\{i : D_i=d_q\}$, $q=1,...,Q$ for $Q$ levels of treatment, and $m_q=\big| \mathcal{M}_q\big|$. We specify a family of PS models $\mathcal{P}=\{\pi^j(x;\alpha^j):j=1,...,J\}$ to estimate $\alpha$ from $f_{D|X}(d|x_i,\alpha)$, using the observed treatment doses $D_i$ and covariates $X_i$. From this we calculate estimates $\hat{\pi}^j(d|X_i;\hat{\alpha}^j)$ for $\pi(d|X_i;\alpha)$. We next postulate a family of OR models $\mathcal{A}=\{a^k(x,d;\beta^k):k=1,...K\}$. Using $\mathcal{P}$ and $\mathcal{A}$ we can now define our MR estimator for a multivalued treatment.

\begin{definition}(Multiply-Robust estimator for multivalued treatment) For possible treatment levels $d_1, ..., d_Q$ we define $\hat{\mu}_{MR}(d_q)$ by
\begin{align}
\hat{\mu}_{MR}(d_q)=\sum_{i \in  \mathcal{M}_q}\hat{w}_{i,d_q}Y_i
\end{align} 
We define the $\hat{w}_{i,d_q}$ as follows:
\begin{enumerate}
\item Postulate multiple OR models $\mathcal{A}=\{a^k(x,d;\beta^k):k=1,...,K\}$ for $a(d,x)$ and estimate the parameters $\hat{\beta}^k$ from the regression coefficients.
\item Assume a p.m.f. $f_{D|X}(d|x_i,\alpha)$ for the GPS and postulate multiple PS models $\mathcal{P}=\{\pi^j(x,d;\alpha^j):j=1,...,J\}$ to estimate the parameters $\hat{\alpha}^j$.  From this, calculate the estimates $\hat{\pi}^j(d_q|X_i;\hat{\alpha}^j)$ for $j=1,...,J$. We can think of this more simply as postulating a family $\mathcal{P'}=\{\hat{\pi}^j(d_q|X_i;\alpha^j):j=1,...,J\}$ of models for $\pi(d_q|x)$, and then estimating the parameters $\hat{\alpha}^j$ from the data.
\item Define $\hat{\theta}^j=\frac{1}{n}\sum_{i=1}^n\hat{\pi}^j(d_q|X_i;\hat{\alpha}^j)$ and $\hat{\eta}^k=\frac{1}{n}\sum_{i=1}^na^k(X_i,d_q;\hat{\beta}^k)$
\item Letting $\hat{\alpha}^T=\{(\hat{\alpha}^1)^T,...,(\hat{\alpha}^j)^T\}$ and $\hat{\beta}^T=\{(\hat{\beta}^1)^T,...,(\hat{\beta}^K)^T\}$, define
\begin{align}\label{eqD20}
\begin{split}
\hat{g}_i(\hat{\alpha},\hat{\beta})=\{&\hat{\pi}^1(d_q|X_i;\hat{\alpha}^1)-\hat{\theta}^1,...,\hat{\pi}^j(d_q|X_i;\hat{\alpha}^j)-\hat{\theta}^j,\\
&a^1(X_i,d_q;\hat{\beta}^1)-\hat{\eta}^1,..., a^K(X_i,d_q;\hat{\beta}^K)-\hat{\eta}^K  \}^T
\end{split}
\end{align}
\item For $i\in \mathcal{M}_q$, set
\begin{align}\label{eqD21}
\hat{w}_{i,d_q}=\frac{1}{m_q}\frac{1}{1+\hat{\rho}^T\hat{g}_i(\hat{\alpha},\hat{\beta})}\bigg/\bigg\{\frac{1}{m_q}\sum_{i\in\mathcal{M}_q}\frac{1}{1+\hat{\rho}^T\hat{g}_i(\hat{\alpha},\hat{\beta})}\bigg\}
\end{align} 
Where $\hat{\rho}^T=(\hat{\rho}_1,...,\hat{\rho}_{J+K})$ is a $(J+K)$-dimensional vector satisfying:
\begin{align}\label{eqD22}
\sum_{i\in\mathcal{M}_q}\frac{\hat{g}_i(\hat{\alpha},\hat{\beta})}{1+\hat{\rho}^T\hat{g}_i(\hat{\alpha},\hat{\beta})}=0
\end{align}
\end{enumerate}
\end{definition}
\begin{theorem}\label{ThmMR}
The MR estimator has the multiply-robust property for estimating the APOs $\mu(d_q)$ in the multivalued treatment case.
\end{theorem}

We start by considering the case in which $\mathcal{P}$ contains a correctly specified PS model, say $\pi^1(d|x;\alpha^1)$ without loss of generality. Let $\alpha_0^1$ denote the true value of $\alpha^1$ such that $\pi^1(d|x;\alpha^1_0)=\pi(d|x)$, and let $\hat{\alpha}^1$ be its estimator that we fit from our PS model. 

Following \citet{Han/Wang:2013}, we define the empirical probability of $(Y_i,X_i)$ conditional on $D_i=d_q$ to be $p_i$ for $i=1,...,m_q$. Here, we index the subjects with $D_i=d_q$ by $1,...,m_q$ without loss of generality. Conditional on $D=d_q$ we then have the following lemma 
\begin{lemma}\label{lemma4}
\begin{align}
\begin{split}
&\mathbb{E}\bigg[\frac{\pi^j(d_q|X;\alpha^j)-\mathbb{E}\{\pi^j(d_q|X;\alpha^j)\}}{\pi(d_q|X)}\bigg|D=d_q\bigg]=0\ \ \ \ \ \ \ \ \ \ j=1,...,J,\\
&\mathbb{E}\bigg[\frac{a^k(X,d_q;\beta^k)-\mathbb{E}\{a^k(X,d_q;\beta^k)\}}{\pi(d_q|X)}\bigg|D=d_q\bigg]=0\ \ \ \ \ \ \ \ \ \ k=1,...,K
\end{split}
\end{align}
\end{lemma}
\begin{proof}The case of $\pi^j(d_q|X;\alpha^j)$ for arbitrary $j$ follows from the fact that
\begin{align*}
0&=\mathbb{E}\bigg(\frac{I_{d_q}(D)}{\pi(d_q|X)}[\pi^j(d_q|X;\alpha^j)-\mathbb{E}\{\pi^j(d_q|X;\alpha^j)\}]\bigg)\\
&=\mathbb{P}(I_{d_q}(D)=1)\mathbb{E}\bigg(\frac{I_{d_q}(D)}{\pi(d_q|X)}\left[\pi^j(d_q|X;\alpha^j)-\mathbb{E}\{\pi^j(d_q|X;\alpha^j)\}\right]\bigg|I_{d_q}(D)=1\bigg)\\
&\ \ \ \ \ +\mathbb{P}(I_{d_q}(D)=0)\mathbb{E}\bigg(\frac{I_{d_q}(D)}{\pi(d_q|X)}\left[\pi^j(d_q|X;\alpha^j)-\mathbb{E}\{\pi^j(d_q|X;\alpha^j)\}\right]\bigg|I_{d_q}(D)=0\bigg)\\
&=\mathbb{P}(I_{d_q}(D)=1)\mathbb{E}\bigg(\frac{\pi^j(d_q|X;\alpha^j)-\mathbb{E}\{\pi^j(d_q|X;\alpha^j)\}}{\pi(d_q|X)}\bigg|I_{d_q}(D)=1\bigg)\\
&=\mathbb{P}(D=d_q)\mathbb{E}\bigg(\frac{\pi^j(d_q|X;\alpha^j)-\mathbb{E}\{\pi^j(d_q|X;\alpha^j)\}}{\pi(d_q|X)}\bigg|D=d_q\bigg)\\
\end{align*}
This then proves the lemma, because the probabilistic assignment condition described in Section two implies that $\mathbb{P}(D=d_q)>0$. The case of $a^k(X,d_q;\beta^k)$ with arbitrary $k$ follows from a similar argument.
\end{proof}
We then proceed as in the original proof. Lemma \ref{lemma4} implies that the most plausible value for $p_i$ should be defined through the constrained optimization problem
\begin{align}\label{eqLM}
\underset{p_1,...,p_{m_q}}{\max}\prod_{i=1}^{m_q}p_i\ \ \ {\rm subject\ to}\ \ \ &p_i\geq0,\ (i=1,...,m_q),\ \ \ \sum_{i=1}^{m_q} p_i=1 \nonumber \\
&\sum_{i=1}^{m_q} p_i\frac{\pi^j(d_q|X_i;\alpha^j)-\hat{\theta}^j}{\pi^1(d_q|X_i;\alpha^1)}=0 \nonumber \\
&\sum_{i=1}^{m_q} p_i\frac{a^k(X_i,d_q;\beta^k)-\hat{\eta}^k}{\pi^1(d_q|X_i;\alpha^1)}=0
\end{align}
The first two constraints make the $p_i$ empirical probabilities, and the last two are the empirical versions of the statements of Lemma \ref{lemma4}, using the fact that $\pi^1(d_q|X;\alpha^1)$ is correctly specified.

Applying Lagrange multipliers gives
\begin{align*}
\hat{p}_i=\frac{1}{m_q}\frac{1}{1+\hat{\lambda}^T\hat{g}_i(\hat{\alpha},\hat{\beta})/\pi^1(d_q|X_i;\alpha^1)}
\end{align*}
Where $\hat{\lambda}^T=(\hat{\lambda}_1,...,\hat{\lambda}_{J+K})$ is a $(J+K)$-dimensional vector satisfying:
\begin{align}\label{eqlam}
\sum_{i=1}^{m_q}\frac{\hat{g}_i(\hat{\alpha},\hat{\beta})/\pi^1(d_q|X_i;\alpha^1)}{1+\hat{\lambda}^T\hat{g}_i(\hat{\alpha},\hat{\beta})/\pi^1(d_q|X_i;\alpha^1)}=0
\end{align}
Comparing this to (\ref{eqD22}), we see that a solution to (\ref{eqD22}) is given by $\hat{\rho}^{\pi}$ with $\hat{\rho}^{\pi}_1=(\hat{\lambda}_1+1)/\hat{\theta}^1$ and $\hat{\rho}^{\pi}_l=\hat{\lambda}_l/\hat{\theta}^1$ for $l=2,...,J+K$. Inserting this into (\ref{eqD21}) gives 
\begin{align}\label{eqWEIGHT}
\hat{w}_{i,d_q}=\frac{1}{m_q}\frac{\hat{\theta}^1/\pi^1(d_q|X_i;\alpha^1)}{1+\hat{\lambda}^T\hat{g}_i(\hat{\alpha},\hat{\beta})/\pi^1(d_q|X_i;\alpha^1)}\bigg/\bigg\{\frac{1}{m_q}\sum_{i=1}^{m_q}\frac{\hat{\theta}^1/\pi^1(d_q|X_i;\alpha^1)}{1+\hat{\lambda}^T\hat{g}_i(\hat{\alpha},\hat{\beta})/\pi^1(d_q|X_i;\alpha^1)}\bigg\}
\end{align}
By the definition of $\hat{p}_i$, we then have that (when $\hat{w}_{i,d_q}$ is given by (\ref{eqD21}) and $\hat{\rho}=\hat{\rho}^{\pi}$)
\begin{align*}
\hat{w}_{i,d_q}=\frac{\hat{p}_i\hat{\theta}^1}{\pi^1(d_q|X_i;\alpha^1)}
\bigg/\sum_{i=1}^{m_q}\frac{\hat{p}_i\hat{\theta}^1}{\pi^1(d_q|X_i;\alpha^1)}
\end{align*}
Using (\ref{eqLM}), we can show that $\sum_{i=1}^{m_q}\frac{\hat{p}_i\hat{\theta}^1}{\pi^1(d_q|X_i;\alpha^1)}=1$. Thus $\hat{w}_{i,d_q}=\hat{p}_i\hat{\theta}^1/\pi^1(d_q|X_i;\alpha^1)$ and our estimator becomes 
\begin{align*}
\hat{\mu}_{MR}(d_q)=\sum_{i=1}^n\frac{I_{d_q}(D_i)\hat{p}_i\hat{\theta}^1Y_i}{\pi^1(d_q|X_i;\alpha^1)}
\end{align*}

It is apparent that the weights $\hat{w}_{i,d_q}$ are positive. We also know as before that for $j=1,...,J$ and $k=1,...,K$, $\hat{\alpha}^j\to \alpha^j_*$ and $\hat{\beta}^k\to \beta^k_*$ in probability Here, $ \alpha^j_*$ and $\beta^k_*$ are the least false values that minimize the corresponding Kullback-Leibler distance between the probability distribution based on the postulated model and the one that generates the data. Since we assumed without loss of generality that $\pi^1(d_q|X_i;\alpha^1)$ was correctly specified, $\alpha_*^1=\alpha_0^1$. Furthermore, $\hat{\theta}^j \to \theta_*^j$ and $\hat{\eta}^k \to \eta_*^k$ in probability under suitable regularity conditions, where $ \theta_*^j=\mathbb{E}\{\pi^j(d_q|X;\alpha_*^j)\}$ and $\eta_*^k=\mathbb{E}\{a^k(X,d_q;\beta^k_*)\}$.

\citet{Han/Wang:2013} show that for their model, $\hat{\lambda}=O_p(n^{-1/2})$. Since we have found $\hat{\lambda}$ in the same way, this holds true for our model as well. Thus, since $\hat{w}_{i,d_q}$ is given by (\ref{eqWEIGHT}), simple algebra gives 
\begin{align*}
\hat{\mu}_{MR}(d_q)=\sum_{i=1}^n\frac{I_{d_q}(D_i)Y_i}{\pi(d_q|X_i)}+o_p(1)
\end{align*}
which tends to $\mu(d_q)$ in probability because $\pi(d_q|X)=\mathbb{E}(I_{d_q}(D)|X)$. 

Hence we have the following result:
\begin{lemma}\label{LmodP}
When $\mathcal{P}$ contains a correctly specified model for $\pi(d|x)$ and (\ref{eqD21}) is evaluated at $\hat{\rho} = \hat{\rho}^{\pi}$, $\hat{\mu}_{MR}(d_q)$ is a consistent estimator of $\mu({d_q})$ as $n \to \infty$
\end{lemma}

We now consider the case that $\mathcal{A}$ contains a correctly specified OR model, say $a^1(x,d;\beta^1)$ without loss of generality. Let $\beta_0^1$ denote the true value of $\beta^1$ such that $a^1(x,d;\beta^1_0)=a(x,d)$. To prove the consistency of $\hat{\mu}_{MR}(d_q)$ in this case, suppose that there exists a solution $\hat{\rho}^a$ to (\ref{eqD22}) which has a finite probability limit $\rho_*^a$ and that (\ref{eqD21}) is evaluated under $\hat{\rho}^a$. By the definition of $\hat{w}_{i,d_q}$ we have that, after some algebra
\begin{align*}
\hat{\mu}_{MR}(d_q)&=\sum_{i\in \mathcal{M}_q}\hat{w}_{i,d_q}\{Y_i-a^1(X_i,d_q;\hat{\beta}^1)\}+\frac{1}{n}\sum_{i=1}^na^1(X_i,d_q;\hat{\beta}^1)\\
&=\sum_{i\in \mathcal{M}_q}\frac{\frac{1}{m_q}\frac{1}{1+\hat{\rho}^T\hat{g}_i(\hat{\alpha},\hat{\beta})}}{\frac{1}{m_q}\sum_{i\in\mathcal{M}_q}\frac{1}{1+\hat{\rho}^T\hat{g}_i(\hat{\alpha},\hat{\beta})}}\{Y_i-a^1(X_i,d_q;\hat{\beta}^1)\}+\frac{1}{n}\sum_{i=1}^na^1(X_i,d_q;\hat{\beta}^1)\\
&=\sum_{i=1}^n\frac{\frac{I_{d_q}(D_i)}{1+\hat{\rho}^T\hat{g}_i(\hat{\alpha},\hat{\beta})}}{\sum_{i=1}^n\frac{I_{d_q}(D_i)}{1+\hat{\rho}^T\hat{g}_i(\hat{\alpha},\hat{\beta})}}\{Y_i-a^1(X_i,d_q;\hat{\beta}^1)\}+\frac{1}{n}\sum_{i=1}^na^1(X_i,d_q;\hat{\beta}^1)\\
&=\bigg[\sum_{i=1}^n\frac{I_{d_q}(D_i)\{Y_i-a^1(X_i,d_q;\hat{\beta}^1)\}}{1+\hat{\rho}^T\hat{g}_i(\hat{\alpha},\hat{\beta})}\bigg]\bigg/\bigg[\sum_{i=1}^n\frac{I_{d_q}(D_i)}{1+\hat{\rho}^T\hat{g}_i(\hat{\alpha},\hat{\beta})}\bigg]+\frac{1}{n}\sum_{i=1}^na^1(X_i,d_q;\hat{\beta}^1)\\
\end{align*}
Then, since $a^1(x,d;\beta^1)$ is the correct OR model, under suitable regularity conditions we have that the final line tends in probability to 
\begin{align}\label{eqAmod}
\overset{p}{\to}\mathbb{E}\bigg[\frac{I_{d_q}(D)\{Y-a(X,d_q)\}}{1+(\rho^a_*)^Tg(\alpha_*,\beta_*;X)}\bigg]\bigg/\mathbb{E}\bigg[\frac{I_{d_q}(D)}{1+(\rho^a_*)^Tg(\alpha_*,\beta_*;X)}\bigg]+\mu(d_q)
\end{align}
where 
\begin{align}
\begin{split}
g(\alpha_*,\beta_*;X)=\{&\hat{\pi}^1(d_q|X;\alpha_*^1)-\theta_*^1,...,\hat{\pi}^j(d_q|X;\alpha_*^j)-\theta_*^j,\\
&a^1(X,d_q;\beta_*^1)-\eta_*^1,..., a^K(X,d_q;\beta_*^K)-\eta_*^K  \}^T
\end{split}
\end{align}
Recalling that $I_{d_q}(D)Y=Y(d_q)$ and thus
\begin{align*}
I_{d_q}(D)\{Y-a(X,d_q)\}=I_{d_q}(D)\{Y(d_q)-a(X,d_q)\}=0
\end{align*}
Hence (\ref{eqAmod}) is equal to $\mu(d_q)$ and so 
\begin{align*}
\hat{\mu}_{MR}(d_q)\overset{p}{\to}\mu(d_q)
\end{align*}
Hence, we have the following result:
\begin{lemma}\label{LmodA}
When $\mathcal{A}$ contains a correctly specified model for $a(x,d)$ and (\ref{eqD21}) is evaluated at $\hat{\rho} = \hat{\rho}^a$ , where $\hat{\rho}^a$ solves (\ref{eqD22}) and has a finite probability limit, $\hat{\mu}_{MR}(d_q)$ is a consistent estimator of $\mu(d_q)$ as $n \to \infty$.
\end{lemma}
Together, Lemmas \ref{LmodP} and \ref{LmodA} prove the result of Theorem \ref{ThmMR}, and thus we see that our constructed MR estimator does indeed have the MR property.

\subsection{Numerical Implementation}
An important considerations for numerical implementation is that in general, (\ref{eqD22}) has multiple roots, and not all of them will lead to MR estimators. The original approach of \citet{Han/Wang:2013} can lead to problems of convergence to an incorrect root, and in this case the estimator will not have the desired MR property. \citet{Han:2014a} and \citet{Han:2014b} solve this problem by utilising an alternative approach. Here, the root is found by solving a convex minimization problem which almost always converges to the correct root. When implementing our estimator in the following simulation study, we use a similar algorithm.
\par 
For each treatment level $d_q$ our estimator is constructed as follows. With $\hat{g}_i(\hat{\alpha},\hat{\beta})$ defined as in (\ref{eqD20}), we define $F_n(\rho)$ by 
\begin{align*}
F_n(\rho) = -\frac{1}{n}\sum_{i\in\mathcal{M}_q}log(1+\hat{\rho}^T\hat{g}_i(\hat{\alpha},\hat{\beta}))
\end{align*}
and let
\begin{align*}
\mathcal{D}_n = \{\rho : 1+\hat{\rho}^T\hat{g}_i(\hat{\alpha},\hat{\beta})>0, i\in\mathcal{M}_q \}
\end{align*}
From here we follow the convex minimization algorithm of \citet{Han:2014a} and \citet{Han:2014b}, where in the construction of $\Delta_1$ and $\Delta_2$ the sums are taken over $i\in\mathcal{M}_q$. We now show that the root found using the method of \citet{Han:2014b} satisfies the requirements for mutiple robustness in the multivalued setting.

\begin{theorem}\label{ThmMR_New}
The MR estimator has the multiply-robust property for estimating the APOs $\mu(d_q)$ in the multivalued treatment case, when the root of (\ref{eqD22}) is found using the convex optimization procedure.
\end{theorem}
As in \citet{Han:2014a}, the additional constraint from the original method is that the weights must be non-negative. As before, we index the subjects with $D_i=d_q$ by $1,...,m_q$ without loss of generality. The estimator for $\hat{\mu}_{MR}(d_q)$ is then $\sum_{i =1}^{m_q}\hat{w}_{i,d_q}Y_i$. For ease of notation, we use $w_i$ in place of $\hat{w}_{i,d_q}$ in what follows. The weights $w_i$ are found by the optimization problem:
\begin{align}\label{eqLM}
\underset{w_1,...,w_{m_q}}{\max}\prod_{i=1}^{m_q}w_i\ \ \ {\rm subject\ to}\ \ \ &w_i\geq0,\ (i=1,...,m_q),\ \ \ \sum_{i=1}^{m_q} w_i=1 \nonumber \\
&\sum_{i=1}^{m_q} w_i\{\pi^j(d_q|X_i;\alpha^j)-\hat{\theta}^j\}=0 \nonumber \\
&\sum_{i=1}^{m_q} w_i\{a^k(X_i,d_q;\beta^k)-\hat{\eta}^k\}=0
\end{align}
Using the method of Lagrange multipliers gives the same weights as in (\ref{eqD21}), where $\hat{\rho}$ still solves (\ref{eqD22}). The additional constraint from the non-negativity of the $w_i$ gives the following constraint on $\hat{\rho}$:
\begin{align}
1+\hat{\rho}^T\hat{g}_i(\hat{\alpha},\hat{\beta})>0, i=1,...,m_q.
\end{align}
Having defined $F_n(\rho)$ and $\mathcal{D}_n$ as above, we have from Lemma 4 that $\textbf{0}$ is inside the convex hull of $ \{\hat{g}_i(\hat{\alpha},\hat{\beta}), i=1,...,m_q \}$ when $m_q$ is sufficiently large. Thus $\mathcal{D}_n$ is a bounded set, with the proof following as in  \citet{Han:2014a}. It is also easy to verify that it is an open convex set, and so is an open bounded convex polytope with its boundary consisting of some or all of the $m_q$ hyperplanes $\{\rho : 1+\hat{\rho}^T\hat{g}_i(\hat{\alpha},\hat{\beta})=0 \}$ $(i=1,...,m_q)$. Furthermore, $F_n(\rho)$ is a strictly convex function on $\mathcal{D}_n$, due to the fact that $-\log(.)$ is a strictly convex function. 
\par
We define
\begin{align*}
\hat{\rho}^{min}=arg\underset{\rho \in \mathcal{D}_n}{\min}F_n(\rho).
\end{align*}
As in the original proof, the properties of  $F_n(\rho)$ and $\mathcal{D}_n$ imply that $\hat{\rho}^{min}$ exists and is a root of (\ref{eqD22}). Therefore, the new estimator is formed by taking $\hat{\rho}=\hat{\rho}^{min}$ in the calculation of the weights.
\par
As $n$ increases,  $F_n(\rho)$ converges to  $F(\rho) = \mathbb{E}[I_{d_q}(D)\{1+{\rho}^T{g}({\alpha_*},{\beta_*})\}]$ in probability pointwise on $\mathcal{D}$ under suitable regularity 
conditions. Here, ${g}({\alpha_*},{\beta_*})$ is as in the proof of Theorem 3, and 
\begin{align*}
\mathcal{D} = \{\rho : 1+{\rho}^T{g}_x({\alpha_*},{\beta_*})>0,\textrm{for any x that is a realization of }  X\mid D=d_q\}
\end{align*}
where ${g}_x({\alpha_*},{\beta_*})$ is ${g}({\alpha_*},{\beta_*})$ with $X=x$. Under some reasonable assumptions of the boundedness of $a^k(x,d_q;\beta_*^k)$ and $\hat{\pi}^j(d_q|x;\alpha_*^j)$, and the finiteness of $\eta_*^k$ and $\theta_*^j$, $\mathcal{D}$ is an open convex set. Furthermore, Lemma 4 implies that $\mathbb{E}[\frac{{g}({\alpha_*},{\beta_*})}{\pi(d_q|X)}|D=d_q]=0$. Thus $\underset{x \sim X \mid D=d_q}{\sup}\{-\tilde{\rho}^T{g}_x({\alpha_*},{\beta_*})\}>0$ for any $\tilde{\rho}$ with $||\tilde{\rho}||=1$, where $x \sim X \mid D=d_q$ means that $x$ is a possible realization of $X \mid D=d_q$. Since ${g}({\alpha_*},{\beta_*})$ is also bounded by assumption, we see that $\mathcal{D}$ is a bounded set, with the proof following as in \citet{Han:2014a}. $F(\rho)$ is a strictly convex function because $-\log(.)$ is, and therefore when it has a minimum in $\mathcal{D}$, the minimizer is unique and the global minimizer.
\par
From this we can see why $\hat{\rho}^{min}$ is the correct root for multiple-robustness. When $\pi(d|x)$ is correctly modelled, by $\pi^1(d|x;\alpha^1)$ without loss of generality, we have that
\begin{align*}
1+(\hat{\rho}^{\pi})^T\hat{g}_i(\hat{\alpha},\hat{\beta}) = \frac{\pi^1(d_q|X_i;\hat{\alpha}^1)}{\hat{\theta}^1}+o_p(1) \ \ \ \ \ (i=1,...,m_q).
\end{align*}
Here, $\hat{\rho}^{\pi}$ is as defined in Lemma 5, and we use the fact that $\hat{\lambda}=O_p(n^{-1/2})$ as shown in the proof of the same lemma, where $\hat{\lambda}$ solves (\ref{eqlam}). Therefore  $\hat{\rho}^{\pi} \in \mathcal{D}_n$ with probability approaching 1. Since  $\hat{\rho}^{\pi}$ solves (\ref{eqD22}), $\frac{\partial F_n(\hat{\rho}^{\pi})}{\partial \rho}=0$. $F_n(\rho)$ is strictly convex, and thus $\hat{\rho}^{\pi}$ is the unique minimizer on $\mathcal{D}_n$ and so $\hat{\rho}^{min}=\hat{\rho}^{\pi}$ with probability approaching 1. The proof of Lemma 5 now follows through using $\hat{\rho}^{min}$ in place of $\hat{\rho}^{\pi}$.
\par
If instead $a(x,d)$ is correctly specified, $\hat{\rho}^{min}$ is still uniquely identified. If we denote the minimizer of $F(\rho)$ in $\mathcal{D}$ by $\rho_*$, then $\rho_*$ is finite due to the boundedness of $\mathcal{D}$. Furthermore, $\hat{\rho}^{min} \overset{p}{\to} \rho_*$, with the proof following through as in \citet{Han:2014a}. The proof of Lemma 6 only requires that the root of (\ref{eqD22}) used has a finite probability limit. Thus the proof follows through using $\hat{\rho}^{min}$ in place of $\hat{\rho}^{a}$. Finally, we see that calculating the weights using $\rho=\hat{\rho}^{min}$ leads to a multiply-robust estimator. Theorem 7 holds because the proofs of lemmas 5 and 6 follow with the new root $\hat{\rho}^{min}$ found using the convex optimization procedure.

\section{Simulations}
In this section we present simulations to demonstrate the robustness of our proposed estimator to problems of confounding and functional form misspecification, which are routinely encountered in causal analyses. 

A multivalued treatment $D$ with 4 levels is assigned as a function of covariates $X$. The outcome $Y$ is quadratic in $D$ with confounding from $X$. This simulation study is adapted from the one performed by \citet{Han/Wang:2013}.
\begin{align*}
X \sim& U[-2.5,2.5]\\
D|X\sim& Bin(3,\pi(X))\\
Y|X,D\sim& \mathcal{N}(a(X,D),\sigma_{Y}^2=2)
\end{align*}
where $logit\{\pi(x)\}=-0.5+0.1x-0.2x^2$ and $a(x,d)=1+2d-0.35d^2+2x+3x^2$

We propose two OR models, one correctly specified ($a^1$) and one incorrectly specified ($a^2$):
\begin{align*}
a^1(x,d;\beta^1)&=\beta^1_0+\beta^1_1d+\beta^1_2d^2+\beta^1_3x+\beta^1_4x^2\\
a^2(x,d;\beta^2)&=\beta^2_0+\beta^2_1d+\beta^2_2x
\end{align*}
We also propose a correctly specified ($\pi^1$) and incorrectly specified ($\pi^2$) model for $\pi(x)$
\begin{align*}
\pi^1(x;\alpha^1)&=(\exp\{\alpha^1_0+\alpha^1_1x+\alpha^1_2x^2\}+1)^{-1}\\
\pi^2(x;\alpha^2)&=1-\exp(-\exp\{\alpha^2_0+\alpha^2_1x+\alpha^2_2e^x\})
\end{align*}
The parameter estimates $\hat{\beta}^k$ are found as the regression coefficients of a generalized linear model for $a(x,d)$, and we can find the estimates $\hat{\alpha}^j$ as the maximizer of the binomial likelihood
\begin{align*}
\prod_{i=1}^n\{\pi^j(X_i;\alpha^j)\}^{D_i}\{1-\pi^j(X_i;\alpha^j)\}^{3-D_i}
\end{align*}
The PS model is then
\begin{align*}
\hat{\pi}(d_q|X)={{N}\choose{d_q}}\hat{\pi}(X)^j(1-\hat{\pi}(X))^{N-d_q}
\end{align*}
where $N=3$, $d_q=0,...,3$ and $\hat{\pi}(X)$ is estimated by one of our two postulated models for $\pi(x)$.

We test the following models:
\begin{enumerate}
\item Four DR models - $\hat{\mu}_{DR\_1010}(d_q)$, $\hat{\mu}_{DRP\_1001}(d_q)$, $\hat{\mu}_{DR\_0110}(d_q)$ and $\hat{\mu}_{DR\_0101}(d_q)$. Here we combine the OR models with the PS models using the inverse weighting approach. With two OR models and two PS models there are four possible combinations. The notation we use is as follows. The first two digits denote which PS model we are using, and the second two digits denote which OR model. For example $\hat{\mu}_{DR\_1001}(d_q)$ denotes the estimator using PS model $\pi^1$ and OR model $a^2$.
\item Five MR models - $\hat{\mu}_{MR\_1101}(d_q)$, $\hat{\mu}_{MR\_1110}(d_q)$, $\hat{\mu}_{MR\_1011}(d_q)$,  $\hat{\mu}_{MR\_0111}(d_q)$ and
\newline $\hat{\mu}_{MR\_1111}(d_q)$. Here we combine the OR models with the PS models using the multiply-robust approach introduced in the previous chapter. For example 
\newline $\hat{\mu}_{MR\_1101}(d_q)$ denotes the estimator using PS models $\pi^1$ and $\pi^2$ and OR model $a^2$.
\end{enumerate}

In each case, point estimates of the APOs for each level of the treatment are presented, as well as empirical variance and bias estimates. In each case these are obtained by calculating the estimate 1000 times on simulated data sets of size 10000.

Table \ref{simres} shows the results for the simulations. We first note that for the DR models, when either the PS or the OR model is correctly specified, the estimators have small biases. However, when neither estimator is correctly specified the bias is several orders of magnitude larger. We also note that the variances for the DR models are small.

We see that the MR estimator is working well compared to the DR estimators, with similarly small biases and with variances similar in size to the DR estimates for either a correctly specified PS or OR model. The one case where the variance appears unusually high is in the estimate of the first treatment level with the $MR\_0111$ estimator which does not include the correctly specified PS model. However, this is still of a comparable size to the variances of the other estimator, and the bias is still very small.

We conclude that the MR estimator is performing as expected, giving an accurate estimate of the APOs when either the family of PS or OR models contains a correctly specified model. Importantly, because we specify two PS and two OR models, we ensure against the case in DR estimation when both models are specified incorrectly. Thus our results confirm that, when there is significant uncertainty about the correct specification of the OR and PS model, it it of value to use a MR estimator. For example, in the case examined in the simulation where we have two OR and two PS models but do not know which is correctly specified, there is a chance of choosing a DR estimator that performs poorly ($DR\_0101$ in this case). However, in the MR setting, we can use an estimator that includes all four models, and the resulting estimator performs well.

\begin{table}[t!]\label{simres}
\caption{Mean Average Potential Outcomes for DR and MR Estimators}
\centering
\begin{tabular}{ |p{2cm}p{2cm}|p{1.5cm}p{1.5cm}p{1.5cm}p{1.5cm}|  }
 \hline
  \multicolumn{2}{|c|}{} &
\multicolumn{4}{|c|}{Treatment level} \\
  \multicolumn{2}{|c|}{} &
$0$ & $1$   &$2$&   $3$\\
 \hline
 \multicolumn{2}{|c|}{Truth}  & 7.253  &8.903 & 9.853&10.103\\[0.2cm]
\hline
\multirow{3}{4em}{DR\_1010} & Av Est & 7.255&8.902&9.853&10.103  \\ 
& Emp Var& 0.008&0.006&0.005&0.008 \\ 
&Bias&0.002&-0.001&0.000&0.000 \\ [0.1cm]
\multirow{3}{4em}{DR\_1001} & Av Est & 7.238&8.903&9.850&10.109 \\ 
& Emp Var& 0.073&0.015&0.012&0.071  \\ 
&Bias& -0.015&0.001&-0.002&0.006 \\ [0.1cm]
\multirow{3}{4em}{DR\_0110} & Av Est & 7.253&8.901&9.853&10.103\\ 
& Emp Var&0.008&0.006&0.005&0.008\\ 
&Bias& 0.000&-0.001&0.000&0.000\\
\multirow{3}{4em}{DR\_0101} & Av Est & 7.010&8.784&9.969&10.479\\ 
& Emp Var&0.071&0.014&0.011&0.061\\ 
&Bias& -0.243&-0.118&0.117&0.377\\
\hline
\multirow{3}{4em}{MR\_1101} & Av Est & 7.251  &8.899  &9.852 &10.098\\ 
& Emp Var& 0.009&0.006&0.005&0.009\\ 
&Bias&-0.002&-0.004&-0.001&-0.005 \\ [0.1cm]
\multirow{3}{4em}{MR\_1110} & Av Est & 7.247&8.899&9.848&10.099\\ 
& Emp Var&0.008&0.005&0.005&0.009 \\ 
&Bias& -0.006&-0.004&-0.005&-0.006\\ [0.1cm]
\multirow{3}{4em}{MR\_1011} & Av Est &7.250&8.905&9.854&10.105\\ 
& Emp Var&0.008&0.006&0.005&0.009 \\ 
&Bias&-0.003&0.002&0.002&0.002\\
\multirow{3}{4em}{MR\_0111} & Av Est & 7.249&8.899&9.851&10.098\\ 
& Emp Var&0.008&0.005&0.005&0.008\\ 
&Bias&-0.004&-0.003&-0.002&-0.005\\
\multirow{3}{4em}{MR\_1111} & Av Est &7.248&8.899&9.850&10.096\\ 
& Emp Var&0.008&0.005&0.005&0.009 \\ 
&Bias&-0.005&-0.004&-0.003&-0.007\\
\hline
\end{tabular}
\end{table}

\section{Conclusions}
In this paper we have built on the work of \citet{Han/Wang:2013}, \citet{Han:2014a} and \citet{Han:2014b} to formulate a multiply robust estimator for causal dose-response estimation. The model is multiply robust in the sense that consistent estimates of average potential outcomes can be obtained under misspecification of all but one of our multiple outcome regression or generalised propensity score models. We have shown that the MR approach can provide a good approximation to linear or nonlinear dose-response functions and is robust to problems of confounding or functional form misspecification. 

\section*{Acknowledgements}
We are grateful to Peisong Han for helpful comments on an earlier draft of this paper.

\section*{Appendix}
\subsection*{Proof of doubly-robust properties for multivalued treatments} 
\begin{proof} 
By the WLLN, we have that
\begin{align*}
\hat{\mu}_{DR}(d_q)\overset{p}{\to} \mathbb{E}\bigg[\frac{I_{d_q}(D)Y}{\pi(d_q|X;\alpha)}-\frac{I_{d_q}(D)-\pi(d_q|X;\alpha)}{\pi(d_q|X;\alpha)}\Psi^{-1}\{m(X,d_q;\beta)\}\bigg]
\end{align*}
Where $\pi(d_q|X;\alpha)$ is the postulated PS model for $\pi(d_q|X)$,  $\Psi^{-1}\{m(X,d_q;\beta)\}$ is the postulated OR model and $\alpha$ and $\beta$ are the true parameters of the PS and OR models respectively. 

Using the SUTVA this quantity is equal to
\begin{align*}
 \mathbb{E}\bigg[\frac{I_{d_q}(D)Y(d_q)}{\pi(d_q|X;\alpha)}-\frac{I_{d_q}(D)-\pi(d_q|X;\alpha)}{\pi(d_q|X;\alpha)}\Psi^{-1}\{m(X,d_q;\beta)\}.\bigg]
\end{align*}
Rearranging, gives
\begin{align} \label{eqA1}
\mathbb{E}\bigg[Y(d_q)\bigg]+ \mathbb{E}\bigg[\frac{I_{d_q}(D)-\pi(d_q|X;\alpha)}{\pi(d_q|X;\alpha)}\bigg\{Y(d_q)-\Psi^{-1}\{m(X,d_q;\beta)\}\bigg\}\bigg]
\end{align}
The first term is our estimand, so all that remains is to show that the second term is 0 if either the PS or OR model is correctly specified. Let us first consider the case where the OR model is correctly specified, so that
\begin{align*}
\Psi^{-1}\{m(X,d_q;\beta)\}=\mathbb{E}(Y|D=d_q,X)
\end{align*}
Then, using the above equality and the law of iterated expectation, the second term of (\ref{eqA1}) becomes
 \begin{align*}
&\mathbb{E}\bigg(\mathbb{E}\bigg[\frac{I_{d_q}(D)-\pi(d_q|X;\alpha)}{\pi(d_q|X;\alpha)}\bigg\{Y(d_q)-\mathbb{E}(Y|D=d_q,X)\bigg\}\bigg|I_{d_q}(D),X\bigg]\bigg)\\
=&\mathbb{E}\bigg(\frac{I_{d_q}(D)-\pi(d_q|X;\alpha)}{\pi(d_q|X;\alpha)}\mathbb{E}\bigg[\bigg\{Y(d_q)-\mathbb{E}(Y|D=d_q,X)\bigg\}\bigg|I_{d_q}(D),X\bigg]\bigg)\\
=&\mathbb{E}\bigg(\frac{I_{d_q}(D)-\pi(d_q|X;\alpha)}{\pi(d_q|X;\alpha)}\bigg\{\mathbb{E}[Y(d_q)|I_{d_q}(D),X]-\mathbb{E}(Y|D=d_q,X)\bigg\}\bigg)\\
=&\mathbb{E}\bigg(\frac{I_{d_q}(D)-\pi(d_q|X;\alpha)}{\pi(d_q|X;\alpha)}\bigg\{\mathbb{E}[Y(d_q)|X]-\mathbb{E}(Y(d_q)|X)\bigg\}\bigg)=0
\end{align*}
where the last line follows from the weak conditional independence condition implied by the no unmeasured confounders assumption, i.e.
\begin{align*}
\mathbb{E}(Y|D=d_q,X)=\mathbb{E}(Y(d_q)|D=d_q,X)=\mathbb{E}(Y(d_q)|X)=\mathbb{E}(Y(d_q)|I_{d_q}(D),X)
\end{align*}
\bigbreak
\noindent Conversely, if the PS model is correctly specified then 
\begin{align*}
\pi(d_q|X;\alpha)=\pi(d_q|X)=\mathbb{P}(D=d_q|X)=\mathbb{E}(I_{d_q}(D)|X)
\end{align*}
Thus, using the above equality and the law of iterated expectation, the second term of (\ref{eqA1}) becomes
\begin{align*}
&\mathbb{E}\bigg(\mathbb{E}\bigg[\frac{I_{d_q}(D)-\pi(d_q|X)}{\pi(d_q|X)}\bigg\{Y(d_q)-\Psi^{-1}\{m(X,d_q;\beta)\}\bigg\}\bigg|Y(d_q),X\bigg]\bigg)\\
=&\mathbb{E}\bigg(\bigg\{Y(d_q)-\Psi^{-1}\{m(X,d_q;\beta)\}\bigg\}\mathbb{E}\bigg[\frac{I_{d_q}(D)-\pi(d_q|X)}{\pi(d_q|X)}\bigg|Y(d_q),X\bigg]\bigg)\\
=&\mathbb{E}\bigg(\bigg\{Y(d_q)-\Psi^{-1}\{m(X,d_q;\beta)\}\bigg\}\frac{\mathbb{E}[I_{d_q}(D)|Y(d_q),X]-\pi(d_q|X)}{\pi(d_q|X)}\bigg)\\
=&\mathbb{E}\bigg(\bigg\{Y(d_q)-\Psi^{-1}\{m(X,d_q;\beta)\}\bigg\}\frac{\mathbb{E}[I_{d_q}(D)|X]-\pi(d_q|X)}{\pi(d_q|X)}\bigg)\\
=&\mathbb{E}\bigg(\bigg\{Y(d_q)-\Psi^{-1}\{m(X,d_q;\beta)\}\bigg\}\frac{\pi(d_q|X)-\pi(d_q|X)}{\pi(d_q|X)}\bigg)=0
\end{align*}
where the second last line follows from the no unmeasured confounders assumption as before.
\bigbreak
\noindent Thus we see that $\hat{\mu}_{DR}(d_q)$ consistently estimates the APO $\mu(d_q)$ when either the OR model or PS model is correctly specified, in the case of multivalued treatment.
\end{proof}

\bibliographystyle{chicago}
\bibliography{MR}

\end{document}